\documentclass[12pt,aps,pra,nofootinbib]{revtex4-2}
\usepackage{amsmath, amsthm, amssymb, mathtools}
\usepackage[colorlinks=true,linkcolor=blue]{hyperref}
\usepackage{graphicx, cleveref} 

\newcommand{\tr}{{\text{tr}}}
\newcommand{\B}{\textbf}

\allowdisplaybreaks
\numberwithin{equation}{section}

\newtheorem{theorem}{Theorem}
\newtheorem{corollary}[theorem]{Corollary}
\newtheorem{lemma}[theorem]{Lemma}
\newtheorem{definition}[theorem]{Definition}

\begin{document}

\title{Clifford entropy}
\author{Gianluca Cuffaro}
\author{Matthew B. Weiss}
\affiliation{University of Massachusetts Boston: QBism Group}
\date{\today}

\begin{abstract}
We introduce the Clifford entropy, a measure of how close an arbitrary unitary is to a Clifford unitary, which generalizes the stabilizer entropy for states. We show that this quantity vanishes if and only if a unitary is Clifford, is invariant under composition with Clifford unitaries, and is subadditive under tensor products. Rewriting the Clifford entropy in terms of the stabilizer entropy of the corresponding Choi state allows us to derive an upper bound: that this bound is not tight follows from considering the properties of symmetric informationally complete sets. Nevertheless we are able to numerically estimate the maximum in low dimensions, comparing it to the average over all unitaries, which we derive analytically. Finally, harnessing a concentration of measure result, we show that as the dimension grows large, with probability approaching unity, the ratio between the Clifford entropy of a Haar random unitary and that of a fixed magic gate gives a lower bound on the depth of a doped Clifford circuit which realizes the former in terms of the latter. In fact, numerical evidence suggests that this result holds reliably even in low dimensions. We conclude with several directions for future research.
\end{abstract}

\maketitle

\section{Introduction}

What supplies the ``magic'' that makes quantum computers go? Long ago, the Gottesman-Knill theorem \cite{gottesman1998heisenbergrepresentationquantumcomputers} demonstrated that the fragment of quantum mechanics consisting only of so-called stabilizer states and Clifford operations is efficiently classically simulable. Indeed, the presence of nonstabilizer states and operations drastically increases the sample cost of simulating expectation values \cite{PhysRevLett.118.090501, PhysRevLett.115.070501}. At the same time, as Clifford operations may be more easily implemented on \emph{quantum} computers, Bravy and Kitaev \cite{kitaev} introduced the \emph{magic state model} of quantum computation: here computational universality is achieved despite a restriction to Clifford operations by the injection of a supply of nonstabilizer states \cite{Veitch_2014}. From a resource-theoretic point of view, therefore, the cost of a quantum computation can be quantified by the number of copies of a standardized nonstablizer state required to realize it. Phrased instead in terms of the unitary which prepares the nonstabilizer (or ``magic'') state (e.g., a $T$-gate), one can consider the number of layers of a $T$-doped Clifford circuit necessary to realize an arbitrary unitary as its cost.

A variety of measures of stabilizerness have been introduced in the literature \cite{PRXQuantum.4.010301, PhysRevLett.124.090505, Peleg_2022}, notably the \emph{stabilizer entropy} of Leone, Oliviero, and Hamma \cite{Leone_2022} which is essentially the unique computable measure for pure states \cite{PhysRevA.110.L040403}. In this paper, we introduce an analogous, efficiently computable measure of the Cliffordness of a unitary: the \emph{Clifford entropy}. In fact, there are many Clifford entropies parameterized by $\alpha\in \mathbb{R}$. We prove that the $\alpha$-Clifford entropies satisfy several desirable properties from a resource theoretic point of view (Theorems \ref{faithfulness}, \ref{clifford_invariance}, \ref{tp_subadditivity}): 
\begin{itemize}
    \item[(i)] Faithfulness: $H_\alpha(U)=0$ iff $U$ is a Clifford unitary.
    \item[(ii)] Clifford invariance: $H_\alpha(C_1 U C_2) = H_\alpha(U)$ for any Clifford unitaries $C_1, C_2$.
    \item[(iii)] Subadditivity under tensor products: $H_\alpha(U \otimes V) \leq H_\alpha(U) + H_\alpha(V)$.
\end{itemize}
By relating the Clifford entropy directly to the stabilizer entropy of the Choi state of the unitary channel (Theorem \ref{choi}), we are able to derive a non tight upper bound on the Clifford entropy (Theorem \ref{not_tight}). In fact, both the bound and the fact cannot it cannot be tight follow from the (non)existence of a symmetric informationally complete (SIC) set with a prescribed set of symmetries \cite{cuffaro2024quantumstatesmaximalmagic,Renes_2004}. It is notable that one can draw conclusions about dynamics from such considerations: it is in line with previous work showing that the study of SICs lies at the heart of finite dimensional quantum theory \cite{DeBrota2020}, and we remark that the effort to resolve the question of SIC existence has led to a remarkable marriage between physics and algebraic number theory \cite{Appleby2025}. Finally, although we cannot establish a tight upper bound analytically, we report the results of numerically optimizing the Clifford entropy for qudit dimensions $d=2,\dots,9$ (Figure \ref{fig:avg_max}), showing that as dimension increases, the average approaches the maximum.

Specializing to the $2$-Clifford entropy, 
\begin{align}
H_2(U) &= 1-\frac{1}{d^6}\sum_{\B{ab}}\left|\tr\Big(D^{\dagger}_{\B{a}}U D_{\B{b}}U^\dagger\Big)\right|^{4},
\end{align}
where $\{D_\B{a}\}$ are a set of Weyl-Heisenberg operators, we are able to establish an operational interpretation of our measure: a $T$-doped Clifford circuit of the form $U=C_1 T C_2 T \dots$ realizing a generic unitary $U$ must have depth at least
\begin{align}
t(U) \ge \frac{H_2(U)}{H_2(T)}
\end{align}
with probability approaching unity as the Hilbert space dimension $d$ grows large (Theorem \ref{operational_interp}). In order to arrive at this result, we characterize the Lipschitz constant of $H_2$ (Corollary \ref{H2_Lipschitz}) as well that of a related function diagnosing the failure of subadditivity under composition (Lemma \ref{GLipschitz}). Moreover, using the Weingarten calculus and a result from \cite{Cepollaro2025StabilizerEntropy}, we explicitly evaluate the average of $H_2$ over the unitary group, showing it goes as $1-O(d^{-2})$ (Theorem \ref{H2_haar}). Finally, we use these results to invoke a known theorem regarding concentration of measure on the unitary group (Theorem \ref{meckes}). We also provide numerical evidence that subadditivity violation is profoundly rare for generic unitaries (Figure \ref{fig:subadditivity}) suggesting that the bound on $t(U)$ is even more reliable that one would naively expect for low $d$.

There are a handful of existing measures of the magic of a quantum operation in the literature. For example, one may compare the Clifford entropy to the \emph{nonstablizing power} of a unitary $U$ \cite{Leone_2022} defined by allowing $U$ to act on all stabilizer states, and averaging their stabilizer entropies: while intuitive in its meaning, the NSP becomes nontrivial to calculate as $d$ grows large. Meanwhile, \cite{Wang_2019} introduces several measures of the magic of a quantum channel related to the nonnegativity of the discrete Wigner function of the channel's Choi state: our work can be seen as complementary, based as it is on the stabilizer entropy of the Choi state.

\section{Preliminaries}
 
\subsection{The Weyl-Heisenberg group}

For any natural number $d$, let $\mathbb{Z}_d=\{0,\dots,d-1\}$ be the set of congruence classes modulo $d$. We define the computational basis of a $d$-dimensional Hilbert space $\mathcal{H}_d$ as the orthonormal set $\{|k\rangle\}_{k\in\mathbb{Z}_{d}}$ such that the  \textit{clock} and \textit{shift} operators act respectively as
\begin{align}
   Z|k\rangle \equiv\omega^k |k\rangle, &&  X|k\rangle\equiv|k+1\rangle,
  \end{align}
where $\omega=e^{2\pi i /d}$ and all arithmetic should be understood modulo $d$. The clock and shift operators satisfy
\begin{align}
    X^d=Z^d=I, && X^kZ^l=\omega^{-kl}Z^lX^k.
\end{align}
In terms of them, we may define displacement operators as
\begin{align}
    D_\textbf{a}\equiv\tau^{a_1a_2}X^{a_1}Z^{a_2}, && \textbf{a}=(a_1,a_2)\in\mathbb{Z}_d\times\mathbb{ Z}_d,
\end{align}
where $\tau=-e^{i\pi/d}$. The displacement operators satisfy $D_{\textbf{a}}^\dagger = D_{-\textbf{a}}$ as well as
\begin{align}\label{HEalgebra}
    D_\textbf{a}D_\textbf{b}=\tau^{[\textbf{a,\textbf{b}]}}D_{\textbf{a}+\textbf{b}}=\omega^{[\textbf{a}, \textbf{b}]}D_{\textbf{b}}D_{\textbf{a}},
\end{align}
where $[\textbf{a},\textbf{b}]=a_1b_2-a_2b_1$ is the symplectic product on $\mathbb{Z}_d\times\mathbb{ Z}_d$ \cite{appleby2005}. 
The \emph{Weyl-Heisenberg group} WH$(d)$ is the group generated by  the displacement operators
\begin{equation}
   \text{WH}(d)\equiv\{\omega^sD_\textbf{a}\}\,,\quad\quad s\in\mathbb{Z}_d,\textbf{a}\in\mathbb{Z}_d\times\mathbb{ Z}_d.
\end{equation}
The order of the group is $d\times d^2$. In what follows, however, we will consider only the quotient group containing $+1$ phases, in which case the order of the group is $d^2$.

If we consider a composite system with Hilbert space $(\mathcal{H}_{d_L})^{\otimes n}$, where $d_L$ is the local dimension and $n$ is the number of qudits (so that $d= d_L^n$), we may define composite displacement operators as the tensor product of displacements on each subsystem \cite{Veitch_2012},
\begin{equation}
    D_{\textbf{a}_1\oplus\dots\oplus\textbf{a}_n}=D_{\textbf{a}_1}\otimes\dots\otimes D_{\textbf{a}_n}.
\end{equation}
 We must then consider symplectic indices valued in $\mathbb{Z}_{d_L}^{2n}$: $\B{a}=\B{a}_1 \oplus \dots \oplus \B{a}_k$. The symplectic product becomes $[\B{a}, \B{b}]=\sum_i [\B{a}_i, \B{b}_i]$, while the phase factors remain $\omega=e^{2\pi i/d_L}$ and $\tau=-e^{\pi i/d_L}$. The multiqudit WH group is then
\begin{equation}
    \text{WH}_n(d_L)\equiv\{\omega^sD_{\textbf{a}_1}\otimes\dots\otimes D_{\textbf{a}_n}\}\,,\quad\quad s\in\mathbb{Z}_{d_L},\textbf{a}_1,\dots,\textbf{a}_n\in\mathbb{Z}_{d_L}\times\mathbb{ Z}_{d_L},
\end{equation}
which similarly has order $d^2$ after quotienting by overall phase.

One may check that the WH operators form a unitary operator basis, satisfying the orthogonality relation $\tr(D_\textbf{a}^\dagger D_\textbf{b})=d\delta_{\textbf{ab}}$. Consequently, a state may be expanded as $\rho = \frac{1}{d}\sum_{\B{a}}\tr(D^\dagger_{\B{a}}\rho )D_{\B{a}}$. We shall call $\mathfrak{c}_{\B{a}}(\rho) =\tr(D_{\B{a}}^\dagger \rho)$ the \emph{characteristic function} of $\rho$.  Similarly, we can define the characteristic function of a linear map $\Phi$ on operators as
\begin{align}
\mathfrak{C}_{\B{ab}}(\Phi) = \frac{1}{d}\tr(D_{\B{a}}^\dagger 	\Phi(D_{\B{b}})),
\end{align}
so that the action of the map on a state can be expressed
\begin{align}
\sum_{\B{b}} \mathfrak{C}_{\B{ab}}(\Phi)\mathfrak{c}_{\B{b}}(\rho) &= \sum_{\B{b}} \frac{1}{d}\tr\Big(D_{\B{a}}^\dagger 	\Phi(D_{\B{b}})\Big)\tr(D_{\B{b}}^\dagger \rho)\\
&=\tr\left(D_{\B{a}}^\dagger \Phi\left(\frac{1}{d}\sum_{\B{b}}	\tr(D_{\B{b}}^\dagger \rho) D_{\B{b}}\right)\right)\\
&=\tr\Big(D_{\B{a}}^\dagger \Phi(\rho)\Big)=\mathfrak{c}_{\B{a}}(\Phi(\rho)).
\end{align}

\subsection{The stabilizer entropy}

The \emph{Clifford group} $\mathcal{C}_n(d_L)$ is defined as the normalizer of the WH group, that is, the set of unitaries which take displacement operators to displacement operators,
\begin{align}
\mathcal{C}_n(d_L) = \{C \in U \ | \ CD_{\B{a}}C^\dagger = \omega^s D_{\B{a}^\prime}\}.
\end{align}
The \emph{stabilizer states} are simultaneous eigenvectors of maximal abelian subgroups of the WH group \cite{Gross2021}. It follows that the Clifford group takes stabilizer states to stabilizer states. For a single qubit, the stabilizer states are the eigenvectors of the Pauli matrices, and the convex hull of the stabilizer states (the stabilizer polytope) forms an octahedron in the Bloch sphere. More generally in prime dimensions, there are $d(d+1)$ stabilizer states which may be organized into $d+1$ mutually unbiased bases.

The \emph{stabilizer entropy} (SE) \cite{Leone_2022,wang_stabilizer_2023} provides a variational characterization of stabilizer states. Given a pure quantum state $|\psi\rangle$, one can  define a probability distribution
\begin{equation}
 \chi_\B{a}(\psi)= \frac{1}{d}|\mathfrak{c}_{\B{a}}(\psi)|^2=\frac{1}{d}|\langle \psi | D_\textbf{a}^\dagger|\psi\rangle|^2\,,\quad\quad D_\textbf{a}\in \text{WH}_n(d_L),
\end{equation}
where $\mathfrak{c}_{\B{a}}(\psi)$ is the characteristic function of the state $\psi \equiv |\psi\rangle\langle \psi |$.
Each component of  $\chi_\B{a}(\psi)$ is clearly non-negative. To see that the components sum to 1, note that the swap operator can be expressed $\mathcal{S}	= \frac{1}{d}\sum_{\B{a}}D_{\B{a}}^\dagger \otimes D_{\B{a}}$ so that 
\begin{align}
\sum_{\B{a}}\chi_\B{a}(\psi)&=\sum_{\B{a}}\frac{1}{d}|\tr(D_{\B{a}}^\dagger \psi)|^2 =\tr\left( \left(\frac{1}{d}\sum_{\B{a}}D_{\B{a}}^\dagger\otimes D_{\B{a}} \right)(\psi \otimes \psi)\right) \nonumber \\
&= \tr(\mathcal{S}\psi \otimes \psi)=\tr(\psi^2)=1
\end{align}
iff $\psi$ is a pure state. We may then consider e.g., the order-$\alpha$ R\'enyi entropy of $\chi(\psi)$,
\begin{equation}
    M_\alpha(|\psi\rangle)\equiv\frac{1}{1-\alpha}\log{\sum_\textbf{a}\chi_\textbf{a}(\psi)^\alpha}-\log{d}.
\end{equation}
The stabilizer entropy $M_\alpha(|\psi\rangle)$ enjoys the following properties:
\begin{itemize}
    \item[(i)] Faithfulness: $M_\alpha(|\psi\rangle)=0$ iff $|\psi\rangle$ is a stabilizer state.
    \item[(ii)] Clifford invariance: $\forall C\in\mathcal{C}_n(d_L): M_\alpha(C|\psi\rangle)=M_\alpha(|\psi\rangle)$.
    \item[(iii)] Additivity under composition: $M_\alpha(|\psi\rangle\otimes|\phi\rangle)=M_\alpha(|\psi\rangle)+M_\alpha(|\phi\rangle)$.
    \item[(iv)] Upper bounded: $M_\alpha(|\psi\rangle)\leq\frac{1}{1-\alpha}\log{\frac{1}{d}\big(1+(d-1)(d+1)^{1-\alpha}\big)}$.
\end{itemize}
It has been shown by one of the authors \cite{cuffaro2024quantumstatesmaximalmagic} that a pure state $|\psi\rangle$ saturates the upper bound on $M_\alpha(|\psi\rangle)$ for $\alpha\ge2$ if and only if $|\psi\rangle$ is a WH-covariant symmetric informationally complete (SIC) fiducial state \cite{Renes_2004}. This can be easily motivated. Let $\{|\psi_\B{a}\rangle = D_\B{a}|\psi\rangle\}$ be a WH-covariant SIC set generated as the orbit of the fiducial $|\psi\rangle$. By the defining property of a SIC, we have 
\begin{align}
    |\langle \psi_\B{a} |\psi_\B{b}\rangle|^2 =  |\langle \psi D_\B{a}^\dagger D_\B{b}|\psi\rangle|^2 =|\langle \psi| D_\B{c}|\psi\rangle|^2 = \frac{d\delta_{\B{ab}}+1}{d+1},
\end{align}
so that the probability distribution $ \chi_\B{c}=\frac{1}{d}|\langle \psi |D_\textbf{c}|\psi\rangle|^2$ must be an almost flat, highly entropic  vector. We note that if $d$ is composite, one may define $M_\alpha(|\psi\rangle)$ with respect to either the single qudit WH group or multiqudit WH groups: the maximum is attained by a SIC fiducial covariant under the same group (if it exists).

Finally, we note that it is often convenient to work instead with linearized stabilizer entropies based on the $\alpha$-Tsallis entropy,
\begin{equation}
    M_{\text{lin}}^{(\alpha)}(\psi)=\frac{1}{\alpha-1}\left(1-d^{\alpha-1}\sum_{\textbf{a}}\chi_\textbf{a}(\psi)^\alpha\right),
\end{equation}
in particular, when $\alpha=2$, we have $M_{\text{lin}}^{(2)}(\psi)=1-\frac{1}{d}\sum_{\B{a}}|\tr(D_{\B{a}}^\dagger \psi)|^4$.

\section{A measure of Cliffordness}

The stabilizer entropy provides a measure of how close an arbitrary state is to a stabilizer state.  We now introduce an analogous measure, the \emph{Clifford entropy}, which quantifies how close an arbitrary unitary is to a Clifford unitary, and show that the Clifford entropy satisfies several key properties natural from a resource theory point of view \cite{QRT}. 

\begin{definition}[$\alpha$-Clifford entropy]
    Fixing a set of Weyl-Heisenberg displacement operators $\{D_\B{a}\}_{\B{a}\in \mathbb{Z}_d^{2n}}$, let $\mathfrak{C}_\B{ab}(U) = \frac{1}{d}\tr(D_\B{a}^\dagger U D_\B{b} U^\dagger)$ be the characteristic function of a unitary channel $U$. Letting $\mathcal{D}_{\B{ab}}(U)=|\mathfrak{C}_\B{ab}(U)|^2$, we define the $\alpha$-Clifford entropy to be
\begin{align}
H_{\alpha}(U)
&= \frac{1}{\alpha-1}\left(1-\frac{1}{d^{2}}\sum_{\B{ab}}\mathfrak{D}_{\B{ab}}(U)^{\alpha}\right)=\frac{1}{\alpha-1}\left(1-\frac{1}{d^{2(\alpha+1)}}\sum_{\B{ab}}\left|\tr\Big(D^{\dagger}_{\B{a}}U D_{\B{b}}U^\dagger\Big)\right|^{2\alpha}\right).
\end{align}
\end{definition}
\noindent To motivate this definition, we first establish a useful lemma: just as $\chi_\B{a}(\psi)=\frac{1}{d}|\mathfrak{c}_\B{a}|^2$ is a probability distribution iff $\psi$ is pure, $\mathfrak{D}_{\B{ab}}(U)$ is bistochastic iff $U$ is unitary.
\begin{lemma}
\label{bistochastic}
The matrix $\mathfrak{D}_{\B{ab}}(U)=	|\mathfrak{C}_\B{ab}(U)|^2$ is bistochastic iff $U$ is unitary.
\end{lemma}
\begin{proof}
	\begin{align}
 \sum_{\B{a}} \mathfrak{D}_{\B{ab}}(U)  &= \frac{1}{d^2} \sum_{\B{a}}\tr\Big(D^{\dagger}_{\B{a}}U D_{\B{b}}U^\dagger\Big)\tr\Big(UD_{\B{b}}^\dagger U^\dagger D_{\B{a}}\Big)\\
 &=\frac{1}{d}\tr\Bigg(UD_{\B{b}}^\dagger U^\dagger \frac{1}{d}\sum_{\B{a}}\tr\Big(D^{\dagger}_{\B{a}}U D_{\B{b}}U^\dagger\Big)D_{\B{a}}\Bigg) \nonumber =\frac{1}{d}\tr\Big(UD_\B{b}^\dagger U^\dagger UD_\B{b}U^\dagger\Big)=1,
 \end{align}
 and by the same argument $ \sum_{\B{b}} \mathfrak{D}_{\B{ab}}(U)=1$.
\end{proof}
\noindent We are therefore justified in considering the order-$\alpha$ Tsallis entropies of the rows or columns of $\mathfrak{D}(U)$: $H_\alpha^{(\B{a}/\B{b})}(U) = \frac{1}{\alpha-1}\big(1-\sum_{\B{b/a}} \mathfrak{D}_{\B{ab}}(U)^\alpha\big)$ for $\alpha \in \mathbb{R}$. Taking both into account, we define the $\alpha$-Clifford entropy of a unitary to be $H_{\alpha}(U) = \frac{1}{2d^4}\sum_{\B{ab}}\big(H_\alpha^{(\B{a})}(U)+H_\alpha^{(\B{b})}(U)	\big)$, which is consistent with our earlier definition. While we specialize in this paper to Clifford entropies based on the Tsallis entropies, one could work instead with R\'enyi entropies.

We now prove that $H_\alpha(U)$ possess three natural properties one would like from a measure of Cliffordness: faithfulness, invariance under Clifford unitaries, and subadditivity under tensor products. 

\begin{theorem}[Faithfulness]
\label{faithfulness}
$H_\alpha(U)=0$ iff	$U$ is Clifford.
\end{theorem}
\begin{proof}
On the one hand, suppose that $U$ is Clifford. Since by definition,  a Clifford unitary permutes displacement operators among themselves up to phase, its characteristic function must be
\begin{align}
\mathfrak{C}_{\B{ab}}(U) = \frac{1}{d}	\tr\big(D_{\B{a}}^\dagger C D_{\B{b}}C^\dagger\big)=\frac{1}{d}e^{i\phi_\B{b}}\tr\Big(D_{\B{a}}^\dagger D_{f(\B{b})}\Big)=e^{i\phi_{\B{b}}}\delta_{\B{a}, f(\B{b})},
\end{align}
for phases $e^{i\phi_\B{b}}$ and symplectic indices $f(\B{b})$.
Thus $\mathfrak{C}(U)$ must be a monomial matrix, having exactly one nonzero entry of modulus 1 per row and per column. It follows that $\mathfrak{D}(U)=|\mathfrak{C}(U)|^2$ must itself be a permutation matrix. Since the rows of a permutation matrix all have a single entry equal to 1, the rest being 0, $\sum_\B{ab} \mathfrak{D}_{\B{ab}}(U)^\alpha = d^2$, and so $H_\alpha(U)=0$. Conversely, if $H_\alpha(U)=0$, since Tsallis entropies are non-negative, $H_\alpha^{(\B{a})}(U)$ and $H_\alpha^{(\B{b})}(U)$ must be zero individually. This occurs when each row (and column) consists of a single 1, the rest of the entries being 0. Since therefore $\mathfrak{D}(U)$ is a permutation matrix, $\mathfrak{C}(U)$ must be a monomial matrix, and $U$ Clifford. 
\end{proof}

\begin{theorem}[Clifford invariance]
\label{clifford_invariance}
For Clifford unitaries $C_1, C_2$, $H_\alpha(C_1 U C_2)=H_\alpha(U)$. 	
\end{theorem}
\begin{proof}
We begin by observing that $\mathfrak{C}(C_1 U C_2)=\mathfrak{C}(C_1)\mathfrak{C}(U)\mathfrak{C}(C_2)$, where $\mathfrak{C}(C_1)$ and $\mathfrak{C}(C_2)$ are monomial matrices.   $\mathfrak{C}(C_1)$ enacts a phased row permutation of $\mathfrak{C}(U)$ and $\mathfrak{C}(C_2)$ enacts a phased column permutation of $\mathfrak{C}(U)$. Thus $\mathfrak{D}(C_1 U C_2)$ is related to $\mathfrak{D}(U)$ by row and column permutations. But $H_\alpha(U)$ is clearly invariant under row and column permutations, and so $H_\alpha(C_1 U C_2)=H_\alpha(U)$.
\end{proof}

\begin{theorem}[Subadditivity under tensor product]
\label{tp_subadditivity}
For $U\in \mathcal{U}(d_1)$ and $V \in \mathcal{U}(d_2)$, 
\begin{align}
H_\alpha(U\otimes V) = H_\alpha(U)+H_\alpha(V)-(\alpha-1)H_\alpha(U) H_\alpha(V).
\end{align}
\end{theorem}
\begin{proof}
It is straightforward to check that $\mathfrak{C}(U \otimes V) = \mathfrak{C}(U) \otimes \mathfrak{C}(V)$ and $\mathfrak{D}(U\otimes V) = \mathfrak{D}(U) \otimes \mathfrak{D}(V)$. Consequently,
 \begin{align}
 H_{\alpha}(U\otimes V) &= \frac{1}{\alpha-1}\left(1- \frac{1}{(d_1 d_2)^2}\sum_{\B{a}_1, \B{a}_2; \B{b}_1, \B{b}_2} \mathfrak{D}_{\B{a}_1, \B{a}_2; \B{b}_1, \B{b}_2}(U\otimes V)^\alpha\right)\\
 &=\frac{1}{\alpha-1}\left(1- \left(\frac{1}{d_1^2}\sum_{\B{a}_1, \B{b}_1} \mathfrak{D}_{\B{a}_1, \B{b}_1}(U)^\alpha \right)\left(\frac{1}{d_2^2}\sum_{\B{a}_2, \B{b}_2}  \mathfrak{D}_{\B{a}_2, \B{b}_2}(V)^\alpha\right)\right)\\
 &=\frac{1}{\alpha-1}\Big(1- \big(1-(\alpha-1)H_\alpha(U)\big)\big(1-(\alpha-1)H_\alpha(V)\big) \Big),
 \end{align} 
 from which the result follows.
\end{proof}
Next, we establish a non-tight upper bound on the Clifford entropy of a unitary by relating it to the stablizer entropy of the unitary channel's Choi state.

\begin{theorem}\label{choi}
The Clifford entropy of a unitary channel $\Phi$ can be directly related to the stabilizer R\'enyi entropy of its corresponding Choi state $\rho_\Phi$ as
\begin{align}
H_{\alpha}(\Phi)
 =\frac{1}{\alpha-1}\big(1-e^{-(\alpha-1)\mathcal{M}_\alpha(\rho_\Phi)}\big).
\end{align}
\end{theorem}
\begin{proof}
Recalling the Choi isomorphism between a quantum channel $\Phi$ and a state $\rho_\Phi$ on a doubled Hilbert space,
\begin{align}
\rho_{\Phi} &= \frac{1}{d}\sum_{ij} \Phi(|i\rangle\langle j|)\otimes |i\rangle\langle j|,
\end{align}
we can consider the characteristic function of $\rho_{\Phi}$ with respect to bipartite WH operators,
\begin{align}
	\mathfrak{c}_{\B{ab}}(\rho_\Phi)&=\tr\big((D_{\B{a}}^\dagger \otimes D_\B{b}^\dagger)\rho_\Phi\big) =\frac{1}{d}\sum_{ij}\tr\Big(D^\dagger_{\B{a}} \Phi(|i\rangle \langle j|)\Big)\tr(D^\dagger_{\B{b}}|i\rangle\langle j|)\\
	&=\frac{1}{d}\tr\Big(D^\dagger_{\B{a}} \Phi\Big(\sum_{ij}\langle j|D^\dagger_{\B{b}}|i\rangle \ |i\rangle \langle j|\Big)\Big)=\frac{1}{d}\tr\big(D^\dagger_{\B{a}}\Phi(D_{\B{b}}^*)\big).
\end{align}
Since for a unitary channel, $\rho_\Phi$ is pure, $\chi_{\B{ab}}(\rho_\Phi)=\frac{1}{d^2}|\mathfrak{c}_{\B{ab}}|^2=\frac{1}{d^4}|\tr\big(D^\dagger_{\B{a}}\Phi(D_{\B{b}}^*)\big)|^2$ will be a probability distribution. Meanwhile, $\mathfrak{D}_{\B{ab}}(\Phi)=\frac{1}{d^2}\left|\tr(D^{\dagger}_{\B{a}}\Phi( D_{\B{b}}))\right|^2$.
 Notice that since $X$ is real valued and $Z$ simply has roots of unity along its diagonal, $D_{a_1,a_2}^* \propto (X^{a_1} Z^{a_2})^*=X^{a_1}Z^{-a_2}\propto D_{a_1, -a_2}$, and so summing over all indices,
\begin{align}
    \sum_\B{ab} \chi_\B{ab}(\rho_\Phi)^\alpha = \frac{1}{d^{2\alpha}}\sum_\B{ab} \mathfrak{D}_{\B{ab}}(\Phi)^\alpha.
\end{align}
It follows that the stabilizer R\'enyi entropy of $\rho_\Phi$ can be expressed as
\begin{align}
    M_\alpha(\rho_\Phi) &= \frac{1}{1-\alpha}\log{\sum_\textbf{a}\chi_\textbf{ab}(\rho_\Phi)^\alpha}-\log{d^2}=\frac{1}{1-\alpha}\log{ \frac{1}{d^{2}}\sum_\B{ab} \mathfrak{D}_{\B{ab}}(\Phi)^\alpha},
\end{align}
from which our claim follows.
\end{proof}

\noindent A direct consequence of this relation is the following corollary.
\begin{corollary}
    A unitary channel achieves maximal Clifford entropy if and only if its Choi state achieves maximal bipartite stabilizer entropy.
\end{corollary}

For example, when $\alpha=2$, using the upper bound on the stabilizer entropy from \cite{cuffaro2024quantumstatesmaximalmagic} we have
\begin{align}
H_{2}(\Phi) &=1-e^{-\mathcal{M}_2(\rho_\Phi)} \leq 1 - \frac{2}{d^2+1},
\end{align}
with equality iff $\rho_\Phi$ is a SIC state covariant under the bipartite WH group. But in fact, this upper bound is \emph{not tight}, as we now show.

\begin{theorem}
    \label{not_tight}
	The upper bound on the Clifford entropy derived from considering the Choi state stabilizer entropy is not tight.
\end{theorem}
\begin{proof}
On the one hand, the Choi state of a unitary channel is maximally entangled. On the other hand, a state which achieves maximal bipartite stabilizer entropy must be a SIC state covariant under two copies of the WH group. It follows that every state of the SIC must be maximally entangled, since local operations leave entanglement invariant, and thus the purity of each reduced subsystem must be $1/d$. But this contradicts Lemma \ref{purity} (a known result \cite{PhysRevA.82.042308} which we reprove in Appendix \ref{avg_purity}): the average purity of a reduced subsystem of SIC on $\mathcal{H}_d \otimes \mathcal{H}_d$ must be $\frac{2d}{d^2+1}$. Since there can be no SIC which is covariant under the bipartite WH group which also contains a maximally entangled state, there is no unitary which achieves the upper bound on Clifford entropy.
\end{proof}
We observe, however, that as $d$ grows large, the average purity of a SIC subsystem goes to 0: thus as the dimension increases, a bipartite SIC state may be approximately maximally entangled, and the corresponding channel approximately unitary. It is therefore possible that the bound could be approximately attained, in the unlikely event that a SIC with the prescribed symmetries exists. All known SICs, however, are covariant only under the single-qudit WH group with the single exception of the Hoggar SIC, which is covariant under the three-qubit WH group. Indeed, \cite{Godsil_2009} provides the following lemma:
\begin{lemma}[Godsil and Roy]
Let $|\phi\rangle \in\mathbb{C}^{2n}$ and $\mathcal{P}_{n}$ the $n$-fold tensor product of the Pauli group (that is, the qubit WH group). Then the set
\begin{equation}
    \{P |\psi\rangle, P\in\mathcal{P}_{n}\}
\end{equation}
forms a SIC only for $n=1$ and $n=3$.
\end{lemma}
\noindent While in \cite{Zhu_2010}, Zhu has proven that a group-covariant SIC in any prime dimension must be covariant with respect to the qudit Weyl-Heisenberg group, a proof that the same is true in composite dimensions is lacking despite numerical evidence in its favor.

Finally, being unable to establish a tight upper bound on the Clifford entropy, we can nevertheless estimate its maximum numerically. In particular, in Figure \ref{fig:avg_max}, we plot the largest value of $H_2$ achieved by numerically optimizing over the unitary group for different values of $d$. In the same plot, we show the average of $H_2$ over the unitary group (Theorem \ref{H2_haar} in Appendix \ref{appendixHaar}), showing that as $d$ grows, the maximum and the average approach each other.

\begin{figure}[h]
    \centering
    \includegraphics[width=0.75\textwidth]{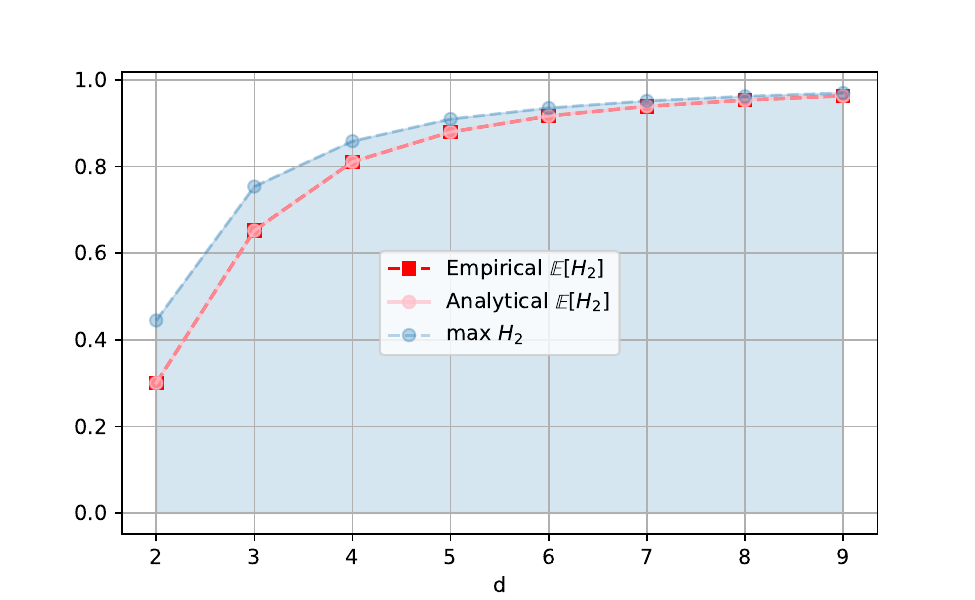}
    \caption{The average of $H_2$ estimated by sampling $25,000$ unitaries is plotted against the analytical value according to Theorem \ref{H2_haar}, with perfect agreement. Above appears an estimate of the largest attainable value of $H_2$ obtained by numerically maximizing $H_2$ over the unitary group using the BFGS optimizer. The optimizer was run $100$ times in each dimension.}
    \label{fig:avg_max}
\end{figure}

\section{Operational interpretation}

Having established the essential properties of the Clifford entropy, we now give our measure an operational interpretation. In particular, we demonstrate a direct link between the Clifford entropy and the so-called\textit{T-count}, the minimum number of non-Clifford gates needed in addition to Clifford gates to realize the unitary in question. To motivate this, suppose that $f(U)$ is a measure of Cliffordness which satisfies additivity under composition, that is, $f(UV) = f(U) + f(V)$. As in a magic state model of quantum computation, if one were to implement a unitary $U$ by a series of Clifford gates interpersed with some fixed non-Clifford gate, e.g., a $T$ gate, then
\begin{align}
f(U) = f(C_1 T C_2T \dots) = f(C_1) + f(T) + f(C_2) + f(T) + \dots = t f(T).
\end{align}
The ratio $t=f(U)/f(T)$ has a immediate interpretation: it is just the $T$-count of a $T$-doped Clifford circuit implementation of the unitary. Relaxing the requirement of strict additivity to \emph{sub}additivity under composition, i.e., $f(UV) \leq f(U) + f(V)$, we would have
\begin{align}
f(U) = f(C_1 T C_2 T \dots) \leq f(C_1) + f(TC_2 T \dots) \leq f(T) + f(C_2 T \dots) \leq t f(T),
\end{align}
so that $t \ge f(U)/f(T)$: the ratio provides a lower bound on the $T$-count necessary to realize $U$. Relaxing even this requirement, one could demand simply that $f(U)$ satisfies subadditivity most of the time, or more specifically that, as $d$ grows large, for a generic unitary $U$, $t \ge f(U)/f(T)$ with probability approaching unity. This latter statement is what we will be able to prove about $H_2(U)$ by appealing to a \emph{concentration of measure} result. In particular, we will make use of Theorem 5.17 of \cite{Meckes2019} (compare \cite{Low2009, Ledoux2005}), about the concentration properties of Haar measure on the classical compact groups.
\begin{theorem}[Meckes]
    \label{meckes}
Given $n_1, \dots, n_k \in \mathbb{N}$, let $X= G_{n_1} \times \cdots \times G_{n_k}$, where for each of the $n_i$, $G_{n_i}$ is one of $\mathcal{SO}(n_i), \mathcal{SO}^-(n_i), \mathcal{SU}(n_i), \mathcal{U}(n_i)$ or $\mathcal{S}p(2n_i)$. Let $X$ be equipped with the $L^2$-sum of Hilbert-Schmidt metrics on the $G_{n_i}$. Suppose that $f:X\rightarrow \mathbb{R}$ is $L$-Lipschitz, and that $\{U_j\in G_{n_j}: 1 \leq j \leq k\}$ are independent, Haar-distributed random matrices. Then for each $t>0$,
\begin{align}
\text{Pr}\Big(f(U_1, \dots, U_k)\ge \mathbb{E}\big[f(U_1, \dots, U_k)\big] + t \Big) \leq e^{-(n-2)t^2/24L^2},
\end{align}
where $n=\min\{n_1, \dots, n_k\}$.
\end{theorem}
\noindent Recall that a function $f(U)$, e.g., on $\mathcal{U}(d)$, is $L$-Lipschitz with respect to the Hilbert-Schmidt norm when $\forall U \neq V \in \mathcal{U}(d): |f(U) - f(V)| \leq L \lVert U -V \rVert_F$. In other words, a Lipschitz function cannot change too quickly. The intuition behind Theorem \ref{meckes} is that Haar measure on high-dimensional compact groups is highly concentrated: in particular, any set of unitaries that is not extremely rare lies close to almost all unitaries. When a function is Lipschitz, this geometric concentration is transferred to the distribution of the function itself. Consequently, typical fluctuations of $f(U)$ about its average are of order $L/\sqrt{n}$.

We will use Theorem \ref{meckes} to bound the probability that subadditivity is violated, and thus safeguard an operational interpretation of the Clifford entropy. In order to do so, however, we will need first to prove several intermediate results, ultimately characterizing the Lipschitz constant not only for the Clifford entropy (Corollary \ref{H2_Lipschitz}) but also for a related function (Lemma \ref{GLipschitz}) which diagnoses the failure of subadditivity. Moreover, we must characterize the average of the Clifford entropy over the unitary group (Theorem \ref{H2_haar}). To see why, consider the following lemma:

\begin{lemma}
\label{subadditivity}
Let $f(U)$ be a function on $\mathcal{U}(d)$, and $g(U,V)= f(UV)-f(U)-f(V)$. If $g$ is $L_g$-Lipschitz, then
\begin{align}
\text{Pr}\big(f \text{ violates subadditivity}\big)&\leq  2\exp \left(-(d-2)\epsilon^2/24L_g^2\right),
\end{align}
where $\epsilon = \mathbb{E}_U[f]$.
\end{lemma}
\begin{proof}
    Let $f(U)$ be a measure of Cliffordness, and define $g(U,V)= f(UV)-f(U)-f(V)$. Then
\begin{align}
   \text{Pr}\big(f \text{ violates subadditivity}\big)&=   \text{Pr}\big(f(UV)\ge f(U)+f(V)\big)\\
   &=\text{Pr}(g \ge 0)\\
   &=\text{Pr}\big(g -\mathbb{E}_{U,V}[g]\ge -\mathbb{E}_{U,V}[g]\big)\\
   &\leq \text{Pr}\big(\big|g-\mathbb{E}_{U,V}[g]\big|\ge -\mathbb{E}_{U,V}[g]\big),
\end{align}
where the last follows from the fact that the former condition is a special case of the latter. But by the invariance of Haar measure
\begin{align}
	\mathbb{E}_{V}[g]&= \mathbb{E}_{V}[f(UV)]-\mathbb{E}_{V}[f(U)]-\mathbb{E}_{V}[f(V)]\\
	&=\mathbb{E}_V[f(V)]-f(U)-\mathbb{E}_{V}[f(V)]\\
	&=-f(U),
\end{align}
so that $\mathbb{E}_{U,V}[g]=-\mathbb{E}_U[f]=-\epsilon$. We conclude
\begin{align}
\text{Pr}\big(f \text{ violates subadditivity}\big)&\leq\text{Pr}\big(\big|g-\mathbb{E}_{U,V}[g]\big|\ge \epsilon \big).
\end{align}
Supposing that $g$ is $L_g$-Lipschitz, we have the result by extending Theorem \ref{meckes} to a two-sided bound.
\end{proof}
\noindent In particular, for the $2$-Clifford entropy, we have the following implication.
\begin{lemma}
    \label{putting_it_together}
\begin{align}
\text{Pr}\big(H_2 \text{ violates subadditivity}\big)&\leq  2\exp \left(-\frac{d(d-2)}{3072\pi^2} + O(1)\right).
\end{align}
\end{lemma}
\begin{proof}
    We employ several results proven in Appendix \ref{appendixC}. By Corollary \ref{H2_Lipschitz}, $H_2(U)$ is Lipschitz with respect to the Hilbert-Schmidt norm on $\mathcal{U}(d)$ with Lipschitz constant $L=4\pi/\sqrt{d}$. Thus by Lemma \ref{GLipschitz}, $G_2(U, V) = H_2(UV) - H_2(U) - H_2(V)$ is Lipschitz with respect to the product Hilbert-Schmidt norm on $\mathcal{U}(d)\times \mathcal{U}(d)$ with Lipschitz constant $L_G = 2\sqrt{2} L = 8\pi\sqrt{2/d}$. Moreover, by Theorem \ref{H2_haar}, $\mathbb{E}_U[H_2(U)]=1-O\left(\frac{1}{d^2}\right)$. The claim then follows directly from Lemma \ref{subadditivity}.
\end{proof}

\noindent Marshalling this result, we can at last connect the Clifford entropy to the $T$-count of a $T$-doped Clifford realization of a unitary $U$.
\begin{theorem}
    \label{operational_interp}
    Consider a set of universal gates which consists in a collection of Clifford unitaries  and a single non-Clifford gate $T$. A $T$-doped Clifford circuit realizing a generic unitary gate $U$ must have depth at least 
    \begin{equation}
        t(U) \geq \frac{H_2(U)}{H_2(T)},
    \end{equation}
   with probability at least $\left[1-2\exp\left(-d(d-2)/3072\pi^2 +O(1)\right)\right]^{H_p(U)/H_p(T)}$. As $d$ grows large, this probability approaches unity.
\end{theorem}
\begin{proof}
    We take $U$ to be realized by a $T$-doped circuit of the form $U= C_1T\dots C_{t(U)}T$ so that $H_2(U)=H_2(C_1T\dots C_{t(U)}T)$. If we assume subadditivity, that is, $H_2(AB) \leq H_2(A)+H_2(B)$ at each step of the circuit, recalling that $H_2$ is invariant under Clifford unitaries and that $H_2(C)=0$ for any Clifford unitary, we have
       \begin{equation}
        t(U)\geq\frac{H_2(U)}{H_2(T)}.
    \end{equation}
    But subadditivity will not always be satisfied. Nevertheless, Corollary \ref{putting_it_together} places a bound the probability of failure. Since we must apply subadditivity at each of the $H_2(U)/H_2(T)$ steps in the circuit, we have
    \begin{equation}
      \text{Pr}(\text{success})\geq\left[1-2\exp\left(-d(d-2)/3072\pi^2 +O(1)\right)\right]^{H_p(U)/H_p(T)}.
    \end{equation}
\end{proof}
\noindent In fact, the interpretation of $H_2(U)/H_2(T)$ as bounding the $T$-count is more reliable in low dimension than our theorem would suggest. Figure \ref{fig:subadditivity} displays the results of numerically estimating the rate of subadditivity failure. Already by $d=3$, the rate is vanishingly small for Haar distributed unitaries.
\begin{figure}[h]
    \centering
    \includegraphics[width=0.75\textwidth]{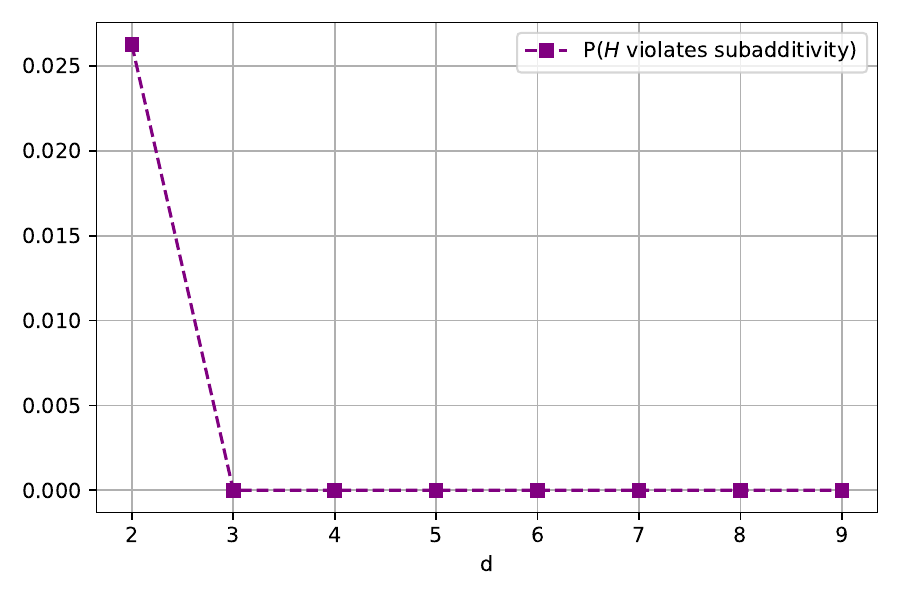}
    \caption{The frequency of subadditivity violation as a function of dimension estimated by sampling $25,000$ pairs of unitaries $U,V$ and calculating the proportion of time that $H_2(UV) \ge H_2(U) + H_2(V)$. This estimate was calculated separately $25$ times and then averaged. After $d=3$, the value is zero up to machine precision.}
    \label{fig:subadditivity}
\end{figure}

\section{Conclusion}

In this work, we have introduced the \emph{Clifford entropy}, a novel measure of how ``magical'' a unitary is. It straightforwardly generalizes the stabilizer entropy for states insofar as both entropies are constructed from probability distributions derived from the expansion of states (or unitaries) in terms of Weyl-Heisenberg operators. We have demonstrated that the Clifford entropy satisfies faithfulness, invariance under Clifford unitaries, and subadditivity under tensor products, and established a not tight upper bound on its value by relating the Clifford entropy of a unitary to the stabilizer entropy of the Choi state of the corresponding channel. On the one hand, the Choi state must be maximally entangled, but on the other hand, to achieve maximal stabilizer entropy (and so maximal Clifford entropy), the Choi state must be a SIC fiducial state covariant under two copies of the WH group. As these two constraints conflict, the bound is not tight, an intriguing instance of the properties of SIC sets implying a result about unitary dynamics. Theorems aside, we provide numerical estimates of the maximum value of the $2$-Clifford entropy in low qudit dimensions, as well as an analytical result characterizing its average of the unitary group. Finally, we employ a concentration of measure theorem to provide an operational spin to our measure: as the dimension grows large, with probability approaching 1, the $T$-count of a $T$-doped Clifford circuit realizing a Haar random unitary $U$ is lower bounded by the ratio $H_2(U)/H_2(T)$. In order to establish this, we first characterize the Lipschitz constants of $H_2$ and a related function which monitors the failure of $H_2$'s subadditivity, characterizing the former in terms of its directional derivative, and the latter in terms of the former. In fact, numerical evidence suggests that our theorem is not as strong as it could be: subadditivity violation is exceedingly rare even in low dimensions.

Going forward, it would be fruitful to relate the Clifford entropy explictly to the nonstabilizing power of \cite{Leone_2022}, and more generally to examine the interplay between the Clifford entropy and the stabilizer entropy: can we make precise the extent to which high Clifford entropy leads to high stabilizer entropy production (on average)? Similarly, how might the Clifford entropy relate to the measures introduced in \cite{Wang_2019}: can they be shown to bound each other? Those measures apply to general quantum channels: thus it would be promising to extend our definition of the Clifford entropy to handle general CPTP maps. We conjecture that when one considers channels whose Choi state have fixed purity, maximal Clifford entropy will be achieved by a higher-rank generalization of a SIC set \cite{Appleby2007} which have a much more flexible existence. In the same vein, can one show that the Clifford entropy is a magic monotone in the sense of \cite{PhysRevA.110.L040403}, where one allows not just Clifford unitaries, but more general Clifford operations? Finally, can one relate the Clifford entropy to measures of scrambling? In a sense, our measure diagnoses scrambling in the discrete phase space: Clifford unitaries permute displacement operators among themselves, while their antipodes according to the Clifford entropy scramble them up as much as possible. We hope that a productive dialogue can open here between phase space notations of scambling and other measures such as out-of-time correlators (OTOCs).

\begin{acknowledgements}
This research was supported in part by the National Science Foundation through grants NSF-2210495 and OSI-2328774. Several key insights in this paper arose in conversation with ChatGPT 5.2.
\end{acknowledgements}

\bibliographystyle{abbrv}
\bibliography{Bibliography}

\appendix 

\section{Average purity of SIC subsystems}
\label{avg_purity}

\begin{lemma}
\label{purity}
The average purity of a reduced subsystem of SIC on $\mathcal{H}_d \otimes \mathcal{H}_d$ is
\begin{align}
\frac{1}{d^4}\sum_i\tr(\rho_i^2)	=\frac{2d}{d^2+1}.
\end{align}
\end{lemma}
\begin{proof}
Because the SIC projectors $\{\Pi_\B{a}\}$ form a minimal unweighted complex projective 2-design \cite{Renes_2004}, they satisfy $\frac{1}{n^2} \sum_i \Pi_i^{\otimes 2}  = \frac{2}{n(n+1)}\Pi_{\text{sym}^2}$ where $n$ is the dimension of the Hilbert space. Taking the partial trace on both sides of subsystems 2 and 4, using the notation $\tr_2(\Pi_i)=\rho_{i,1}$, we find $\frac{1}{n^2} \sum_i \rho_{i,1} \otimes \rho_{i,3}=\frac{2}{n(n+1)}\tr_{24}(\Pi_{\text{sym}^2} )$.
But $\rho_{i, 1}=\rho_{i,3}$, so that multiplying from the left by a swap operator $\mathcal{S}$ and taking the trace yields,
\begin{align}
    \label{avg_pur_eq}
\frac{1}{n^2}\sum_i\tr(\mathcal{S}\rho_{i,1}^{\otimes 2})=\frac{1}{n^2}\sum_i\tr(\rho_{i,1}^2)=	\frac{2}{n(n+1)}\tr(\mathcal{S} \tr_{24}(\Pi_{\text{sym}^2} )).	
\end{align}
Since for a SIC in $\mathcal{H}_d \otimes \mathcal{H}_d$, $n=d^2$, we have $\Pi_{\text{sym}^2} = \frac{1}{2}(I_{d^4}+\mathcal{S})$, so that  $\tr_{24}(\Pi_{\text{sym}^2})=\frac{1}{2}\Big(d^2 I_d \otimes I_d+d\sum_{k,l} |l\rangle\langle k| \otimes |k\rangle\langle l|\Big)$ which yields $\tr(\mathcal{S}	\tr_{24}(\Pi_{\text{sym}^2}))=d^3$. We conclude that the RHS of Eq. \!\ref{avg_pur_eq} is $\frac{2d^3}{d^2(d^2+1)}=\frac{2d}{d^2+1}$ from which the result follows.
\end{proof}

\section{Haar average}\label{appendixHaar}

\begin{theorem}
\label{H2_haar}
    The Haar average of $H_2(U)$ over $\mathcal{U}(d)$ is
\begin{equation}
    \mathbb{E}_U[H_2(U)]=\begin{dcases}
        1 - \frac{3 \left( d^4 - 4d^2 + 2 \right)}{\, d^4\,(d^2 - 4) }, &\text{odd }d\\
        1-\frac{3 \left( d^4 - 10d^2 + 10 \right)}{\, d^2\,(d^2 - 9)\, (d^2 - 1)},& \text{even }d\\
        1 - \frac{4 \left( d^4 - 9 d^2 + 6 \right)}{\, d^4\,(d^2 - 9)},& d=2^n\text{ (qubits)}.
    \end{dcases}
\end{equation}
\end{theorem}
\begin{proof}
The 2-Clifford entropy can be expressed
\begin{align}
H_2(U) &= 1-\frac{1}{d^6}\sum_{\B{ab}}\left|\tr\Big(D^{\dagger}_{\B{a}}U D_{\B{b}}U^\dagger\Big)\right|^{4}.
\end{align}
Defining the operator
\begin{equation}
    Q\equiv\frac{1}{d^2}\sum_{\textbf{a}}(D_\textbf{a}\otimes D_\textbf{a}^\dagger)^{\otimes2} = Q^\dagger,
\end{equation}
it is straightforward to see that the $H_2(U)$ can be written compactly as
\begin{equation}
    H_2(U)=1-\frac{1}{d^2}\tr(QU^{\otimes4}QU^{\dagger\otimes4}).
\end{equation}
\noindent We can then employ Weingarten calculus \cite{Collins2006,Collins_2022} to show that
\begin{align}
    \mathbb{E}_U[H_2(U)]&=1-\frac{1}{d^2}\tr\left(Q\mathbb{E}_U[U^{\otimes4}QU^{\dagger\otimes4}]\right)\\&=1-\frac{1}{d^2}\tr\left(Q\sum_{\pi,\sigma\in S_4}\text{Wg}_{\pi\sigma}\,\tr(QT_\pi)\,T_\sigma\right)\\&=1-\frac{1}{d^2}\sum_{\pi,\sigma\in S_4}\text{Wg}_{\pi\sigma}\tr(QT_\sigma)\tr(QT_\pi),
\end{align}
where $\text{Wg}_{\pi\sigma}$ is the Weingarten matrix defined as the (pseudo)inverse of the Gram matrix $G_{\pi\sigma}\equiv d^{l(\pi^{-1}\sigma)}$. Here $l(\pi)$ is the number of cycles of the permutation $\pi\in S_4$. Appendix D in \cite{Cepollaro2025StabilizerEntropy} provides expressions for $\tr(Q T_\pi)$ for all $\pi \in S_4$ from which the average of $H_2(U)$ can be explicitly computed. In particular, we note that $\mathbb{E}_U[H_2(U)]=1-O\left(\frac{1}{d^2}\right)$.
\end{proof}

\section{Bounding subadditivity violation}\label{appendixC}
We begin with two preliminary lemmas.
\begin{lemma}
    \label{dt_abs_bound}
    Let $f:\mathbb{R}\rightarrow \mathbb{C}$ be differentiable. Then
    \begin{align}
\left|\frac{d}{dt}\big|f(t)\big|^4\right| \leq 4 |f(t)|^3 \left|\frac{d}{dt}f(t)\right|.
    \end{align}
\end{lemma}
\begin{proof}
Letting $r(t) = |f(t)|^2$, we have
\begin{align}
\frac{d}{dt}|f(t)|^4 &= \frac{d}{dt} r(t)^2 = 2r(t)\frac{d}{dt} r(t).
\end{align}
But
\begin{align}
\frac{d}{dt} r(t) &= \frac{d}{dt} f(t) f(t)^* = f(t)\frac{d}{dt} f(t)^* + f(t)^* \frac{d}{dt} f(t) = 2 \Re\Big( f(t)^* \frac{d}{dt}f(t)\Big),
\end{align}
since the two terms are complex conjugates of each other. Thus
\begin{align}
\frac{d}{dt}|f(t)|^4 = 4|f(t)|^2 \Re\Big( f(t)^* \frac{d}{dt}f(t)\Big).
\end{align}
Taking the absolute value of both sides, and using the fact that for $z\in \mathbb{C}$, $|\Re(z)|\leq |z|$, we have
\begin{align}
\left| \frac{d}{dt}|f(t)|^4 \right| \leq 4 |f(t)|^2 \Big| f(t)^* \frac{d}{dt}f(t)\Big| \leq 4 |f(t)|^2  |f(t)^* | \left|\frac{d}{dt}f(t)\right| = 4 |f(t)|^3 \left|\frac{d}{dt}f(t)\right|.
\end{align}
\end{proof}

\begin{lemma} 
    \label{diff_identity}
    Let $U(t) = e^{tK}U$. Then for any $X$,
\begin{align}
\frac{d}{dt}\Big(U(t) X U(t)^\dagger\Big) = [K, U(t) X U(t)^\dagger].
\end{align}
\end{lemma}
\begin{proof}
\begin{align}
\frac{d}{dt}\Big(U(t) X U(t)^\dagger\Big) &= \Big(\frac{d}{dt}U(t)\Big) X U(t)^\dagger + U(t) X \frac{d}{dt}U(t)^\dagger.
\end{align} 
But $\frac{d}{dt}U(t) = \frac{d}{dt}e^{tK}U=Ke^{tK}U=K U(t)$. Similarly, $\frac{d}{dt}U(t)^\dagger = \frac{d}{dt} U^\dagger e^{-tK} = -U^\dagger K e^{-tK} = -U^\dagger e^{-tK}K = -U(t)^\dagger K $. Thus
\begin{align}
\frac{d}{dt}\Big(U(t) X U(t)^\dagger\Big)= K U(t) X U(t)^\dagger - U(t) X U(t)^\dagger  K =  [K, U(t) X U(t)^\dagger].
\end{align}
\end{proof}

\noindent Armed with these results, we now bound the magnitude of the directional derivatives of $H_2$. Recall that the directional derivative at $U\in \mathcal{U}(d)$ in the tangent direction $KU$ can be expressed $\frac{d}{dt}H_2(e^{tK}U)\Big|_{t=0}$. To see this, suppose we have a smooth curve $\gamma(t)$ such that $\gamma(0)=U$. Now since $\gamma(t)\in \mathcal{U}(d)$, we have $\gamma(t)\gamma(t)^\dagger = I$: differentiating both sides yields
\begin{align}
0 &= \frac{d}{dt}\big(\gamma(t)\gamma(t)^\dagger \big) \Longrightarrow \dot{\gamma}(t)\gamma(t)^\dagger=-\gamma(t)\dot{\gamma}(t)^\dagger ,
\end{align}
so that $K =  \dot{\gamma}(t)\gamma(t)^\dagger$ is anti-Hermitian, and $\dot{\gamma}(t) = K\gamma(t)$. We thus see that a vector $\dot{\gamma}(t)$ in the tangent space at $U=\gamma(t)$ can be written in the form $KU$: in particular, at the identity, $\dot{\gamma}(t) =K$ must be anti-Hermitian itself. If we consider in particular $\gamma(t)=e^{tK}U$, so that $\gamma(0)=U$ and $\dot{\gamma}(0)=KU$, the derivative $\frac{d}{dt}H_2(e^{tK}U)\Big|_{t=0}$ is precisely the directional derivative at $U$ in the tangent direction $KU$. Ultimately, we will relate the following bound on the magnitude of $H_2$'s directional derivatives to a bound on its Lipschitz constant.

\begin{theorem}
    \label{derivative_bound}
    For all $U \in \mathcal{U}(d)$ and $K^\dagger=-K$ with $\lVert K \rVert_F=1$, 
\begin{align}
\left| \frac{d}{dt}H_2(e^{tK}U)\Big|_{t=0}\right| \leq \frac{8}{\sqrt{d}}
\end{align}
\end{theorem}
\begin{proof}
The 2-Clifford entropy can be expressed
\begin{align}
H_2(U) &= 1- \frac{1}{d^2}\sum_{\B{ab}} \mathfrak{D}_\B{ab}(U)^2=1-\frac{1}{d^2}\sum_{\B{ab}} |\mathfrak{C}_{\B{ab}}(U)|^4,
\end{align}
where $\mathfrak{C}_{\B{ab}}(U) = \frac{1}{d}\tr(D_\B{a}^\dagger U D_\B{b} U^\dagger)$, and $\mathfrak{D}(U) = |\mathfrak{C}(U)|^2$ is bistochastic (Lemma \ref{bistochastic}). Let $K$ satisfy $K^\dagger=-K$, $\lVert K \rVert_F = 1$ and let  $U(t) = e^{tK}U$. Writing $H_2(t) = H_2(U(t))$, first using the triangle inequality, and then Lemma \ref{dt_abs_bound}, 
\begin{align}
\left|\frac{d}{dt}H_2(t)\right| &=\frac{1}{d^2}\left|\sum_{\B{ab}}\frac{d}{dt}|\mathfrak{C}_{\B{ab}}(t)|^4\right| \leq \frac{1}{d^2}\sum_{\B{ab}} \left| \frac{d}{dt}|\mathfrak{C}_{\B{ab}}(t)|^4\right|\leq \frac{4}{d^2}\sum_{\B{ab}} |\mathfrak{C}_\B{ab}(t)|^3 \left| \frac{d}{dt}\mathfrak{C}_\B{ab}(t)\right|.
\end{align}
By Lemma \ref{diff_identity},
\begin{align} 
\left|\frac{d}{dt}\mathfrak{C}_\B{ab}(t)\right| &=\frac{1}{d}\left|\tr\left(D_\B{a}^\dagger \frac{d}{dt}\big(U(t) D_\B{b} U(t)^\dagger\big) \right)\right| = \frac{1}{d}\left|\tr\Big(D_\B{a}^\dagger \big[K, U(t)D_\B{b}U(t)^\dagger\big]\Big)\right|\\
&= \frac{1}{d}\left|\tr\Big( K\big[ U(t) D_\B{b} U(t)^\dagger, D_\B{a}^\dagger\big]\Big)\right|\leq \frac{1}{d}\lVert K \rVert_F \left\lVert \big[ U(t) D_\B{b} U(t)^\dagger, D_\B{a}^\dagger\big]\right\rVert_F\\
&\leq \frac{1}{d}\big(\lVert U(t) D_\B{b} U(t)^\dagger D_\B{a}^\dagger\rVert_F + \lVert D_\B{a}^\dagger U(t) D_\B{b} U(t)^\dagger \rVert_F \big)\leq \frac{2}{\sqrt{d}},
\end{align}
where we've used the Cauchy-Schwartz inequality, the fact that $\lVert K \rVert_F=1$, the triangle inequality, and the fact that $\lVert U \rVert_F = \sqrt{d}$ for any unitary. Thus
\begin{align}
\left|\frac{d}{dt}H_2(t)\right|  \leq \frac{8}{d^2\sqrt{d}}\sum_{\B{ab}} |\mathfrak{C}_\B{ab}(t)|^3 \leq \frac{8}{\sqrt{d}},
\end{align}
where the last follows from the fact that since $|\mathfrak{C(t)}|^2$ is bistochastic, the sum must be less than $d^2$. 
\end{proof}

\noindent We now show how to translate a bound on the magntitude on the directional derivatives into a bound on the Lipschitz constant.

\begin{lemma}
    \label{deriv_justification}
Let $f:\mathcal{U}(d)\rightarrow \mathbb{R}$ be differentiable, and suppose that for all $U\in \mathcal{U}(d)$ and $K$ such that $K^\dagger = -K$, $\lVert K \rVert_F=1$, there exists some $C>0$ such that
\begin{align}
\left| \frac{d}{dt}f\big(e^{tK}U\big)\Big|_{t=0}\right| \leq C.
\end{align}
Then
\begin{align}
|f(U)-f(V)| \leq \frac{\pi}{2}C\lVert U-V \rVert_{F},
\end{align}
that is, $f$ is $\pi C /2 $-Lipschitz with respect to the Hilbert-Schmidt norm.
\end{lemma}
\begin{proof}
Let $\gamma: [0,1]\rightarrow \mathcal{U}(d)$ be some piecewise $C^1$ curve such that $\gamma(0)=V$ and $\gamma(1)=U$. By the fundamental theorem of calculus,
\begin{align}
f(U)-f(V) = \int_0^1 \frac{d}{dt}f(\gamma(t))dt \Longrightarrow |f(U)-f(V)| \leq \int_0^1 \left| \frac{d}{dt}f(\gamma(t))dt \right|,
\end{align}
where we've applied the triangle inequality. Now since $\gamma(t)\in \mathcal{U}(d)$, we have $\gamma(t)\gamma(t)^\dagger = I$: differentiating both sides yields
 $K =  \dot{\gamma}(t)\gamma(t)^\dagger$ is anti-Hermitian. Moreover, by the unitary invariance of the Frobenius norm $\lVert K(t)\rVert_F = \lVert  \dot{\gamma}(t)\gamma(t)^\dagger\rVert_F=\lVert \dot{\gamma}(t)\rVert_F$. Thus let
\begin{align}
\hat{K}(t) =\frac{K(t)}{\lVert K(t)\rVert_F} = \frac{\dot{\gamma}(t)\gamma(t)^\dagger}{\lVert \dot{\gamma}(t)\rVert_F},
\end{align}
so that 
\begin{align}
\lVert \dot{\gamma}(t)\rVert_F \frac{d}{d\tau}\Big(e^{\tau \hat{K}(t)}\gamma(t) \Big)\Bigg|_{\tau=0}=\lVert \dot{\gamma}(t)\rVert_F \hat{K}(t)\gamma(t) = \lVert \dot{\gamma}(t)\rVert_F \frac{\dot{\gamma}(t)\gamma(t)^\dagger}{\lVert \dot{\gamma}(t)\rVert_F} \gamma(t) =  \dot{\gamma}(t).
\end{align}
By the chain rule, therefore, 
\begin{align}
\frac{d}{dt}f(\gamma(t)) = \lVert \dot{\gamma}(t)\rVert_F \frac{d}{d\tau}f\big( e^{\tau\hat{K}(t)}\gamma(t)\big)\Big|_{\tau=0},
\end{align}
where $\hat{K}(t)=-\hat{K}(t)$ and $\lVert \hat{K}(t)\rVert_F=1$. We now assume that for all such $K$ and for all $U\in \mathcal{U}(d)$, there exists a $C>0$ such that
\begin{align}
\left| \frac{d}{dt}f\big(e^{tK}U\big)\Big|_{t=0}\right| \leq C.
\end{align}
We have immediately that 
\begin{align}
\left|\frac{d}{dt}f(\gamma(t))dt \right| \leq C \lVert \dot{\gamma}(t)\rVert_F,
\end{align}
so that 
\begin{align}
|f(U)-f(V)| \leq C \int_0^1  \lVert \dot{\gamma}(t)\rVert_F = C \cdot \text{Length}(\gamma).
\end{align}
Let us take $\gamma$ to be a geodesic with respect to the Hilbert-Schmidt metric on $\mathcal{U}(d)$: by definition, this will be the shortest such curve. Denoting by $d_g(U,V)$ the geodesic distance, we have $|f(U)-f(V)| \leq C d_g(U,V)$. Lemma 1.3 of \cite{Meckes2019}, however, tells us that for $U,V\in \mathcal{U}(d)$, 
\begin{align}
d_\text{HS}(U,V) \leq d_g(U,V) \leq \frac{\pi}{2} d_\text{HS}(U,V),
\end{align}
where $d_\text{HS}= \lVert U-V\rVert_F=\sqrt{\tr((U-V)^\dagger (U-V))}$. We conclude that
\begin{align}
\forall U,V \in \mathcal{U}(d): |f(U)-f(V)| \leq \frac{\pi}{2} C \lVert U- V\rVert_F,
\end{align} 
as desired.
\end{proof}

\noindent Putting everything together, we can bound the Lipschitz constant of $H_2$.

\begin{corollary}
    \label{H2_Lipschitz}
By Theorem \ref{derivative_bound},  for all $U \in \mathcal{U}(d)$ and $K^\dagger=-K$ with $\lVert K \rVert_F=1$, 
\begin{align}
\left| \frac{d}{dt}H_2(e^{tK}U)\Big|_{t=0}\right| \leq \frac{8}{\sqrt{d}}.
\end{align}
Applying Lemma \ref{deriv_justification}, we conclude
\begin{align}
|H_2(U) - H_2(V)| \leq \frac{4\pi}{\sqrt{d}}\lVert U - V \rVert_F,
\end{align}
that is, $H_2(U)$ is Lipschitz with respect to the Hilbert-Schmidt norm on $\mathcal{U}(d)$ with Lipschitz constant $L \leq 4\pi/\sqrt{d}$.
\end{corollary}

\noindent Finally, we show how the Lipschitz constant of a function $f$ on the unitary group can be related to the Lipschitz constant of a function $g$ which diagnoses the failure of subadditivity of $f$. 
\begin{lemma}
\label{GLipschitz}
Let $f(U)$ be an $L$-Lipschitz function on $\mathcal{U}(d)$ and let $g(U,V) = f(UV) - f(U) - f(V)$.  The $g(U,V)$ is Lipschitz with Lipschitz constant $L_g=2\sqrt{2}L$ with respect to the product Hilbert-Schmidt norm $\lVert (U,V)\rVert_F = \sqrt{\lVert U \rVert_F^2 + \lVert V \rVert_F^2}$ on $\mathcal{U}(d)\times \mathcal{U}(d)$. 
\end{lemma}
\begin{proof} Using first the triangle inequality and then the fact that $f(U)$ is $L$-Lipschitz,
\begin{align}
&|g(U_1, V_1)- g(U_2, V_2)| \nonumber \\
&= |f(U_1V_1)-f(U_1)-f(V_1)-f(U_2V_2)+f(U_2)+f(V_2)|\\
&\leq |f(U_1V_1)-f(U_2V_2)|+|f(U_1)-f(U_2)|+|f(V_1)-f(V_2)|\\
&\leq L\lVert U_1 V_1 - U_2V_2\rVert_F+L\lVert U_1  - U_2\rVert_F+L\lVert V_1  - V_2\rVert_F.
\end{align}
But from the triangle inequality again and the unitary invariance of the Hilbert-Schmidt norm,
\begin{align}
\lVert U_1 V_1 - U_2V_2\rVert_F&=\lVert U_1V_1 - U_1V_2+U_1V_2-U_2V_2\rVert_F\\
&=\lVert U_1(V_1-V_2)+(U_1-U_2)V_2\rVert_F\\
&\leq \lVert U_1(V_1-V_2)\rVert_F + \lVert (U_1-U_2)V_2\rVert_F\\
&=\lVert V_1-V_2\rVert_F + \lVert U_1 - U_2\rVert_F.
\end{align}
We conclude that
\begin{align}
	|g(U_1, V_1)- g(U_2, V_2)|&\leq 2L\big(\lVert U_1 - U_2\rVert_F+\lVert V_1-V_2\rVert_F \big)\\
	&\leq 2\sqrt{2}L\big(\lVert U_1 - U_2\rVert^2_F+\lVert V_1-V_2\rVert^2_F \big)^{1/2}\\
	&=2\sqrt{2}L\lVert(U_1,V_1)-(U_2,V_2) \rVert_F,
\end{align}
where in the second line we've used the Cauchy-Schwartz inequality in the form $a+b \leq \sqrt{2}(a^2+b^2)^{1/2}$.
\end{proof}

\end{document}